\documentclass[letterpaper, 10 pt, conference]{ieeeconf}  
\IEEEoverridecommandlockouts                              
\usepackage{cite}
\usepackage{amsmath,amssymb,amsfonts}
\usepackage{graphicx}
\usepackage{textcomp}
\usepackage{xcolor}
\usepackage{todonotes}
\def\BibTeX{{\rm B\kern-.05em{\sc i\kern-.025em b}\kern-.08em
		T\kern-.1667em\lower.7ex\hbox{E}\kern-.125emX}}	
\usepackage[T1]{fontenc} 
\usepackage{tikz} 

\usepackage{algorithm,algpseudocode}
\usepackage{graphicx}
\usepackage[utf8]{inputenc}

\usepackage{defs}

\begin{document}
	
	\title{\LARGE \bf
		Overall Complexity Certification of a Standard Branch and Bound Method for Mixed-Integer Quadratic Programming} 
	
	\author{Shamisa Shoja, Daniel Arnström, and Daniel Axehill 
		\thanks{S. Shoja, D. Arnström, and D. Axehill are
                  with the Division of Automatic Control, Department of
			Electrical Engineering, Linköping University, Sweden. (Email: 
			{\tt\small \{shamisa.shoja, daniel.arnstrom, daniel.axehill\}@liu.se})}
	}
	
	\maketitle
	\thispagestyle{empty}
	\pagestyle{empty}
	
	\begin{abstract}
          This paper presents a method to certify the computational complexity of a standard Branch and Bound method for solving Mixed-Integer Quadratic Programming (MIQP) problems defined as instances of a multi-parametric MIQP. Beyond previous work, not only the size of the binary search tree is considered, but also the exact complexity of solving the relaxations in the nodes by using recent result from exact complexity certification of active-set QP methods. With the algorithm proposed in this paper, a total worst-case number of QP iterations to be performed in order to solve the MIQP problem can be determined as a function of the parameter in the problem. An important application of the proposed method is Model Predictive Control for hybrid systems, that can be formulated as an MIQP that has to be solved in real-time. The usefulness of the proposed method is successfully illustrated in numerical examples.

	\end{abstract}
	
	
		\newtheorem{lemma}{Lemma}
		\newtheorem{corollary}{Corollary}
		\newtheorem{theorem}{Theorem}
		\newtheorem{assumption}{Assumption}
		\newtheorem{remark}{Remark}
		\newtheorem{definition}{Definition}		
			
        \section{Introduction}
The main motivation for this work is model predictive control (MPC) for discrete-time hybrid systems. MPC is a model-based control strategy with the aim to design an optimal controller for multi-variable constrained systems, \cite{borrelli2017predictive}. 
In MPC, a finite-horizon optimal control problem is solved at each sampling time, starting at the current state.  Measurements are used to update the problem at the next sample, and the optimization is repeated over the shifted horizon.  

Hybrid systems arise naturally in applications where physical principles interact with discrete events. There are different models for hybrid systems, the one that will be considered in this work is Mixed-Logical Dynamical (MLD) systems, \cite{bemporad1999control}. 
In these models, both real-valued variables and binary-valued variables exist in the optimization problem to be solved online in each sample. The result is an optimization problem in the form of a Mixed-Integer Quadratic Program (MIQP), which are non-convex problems known to be $\mathcal{NP}$-hard \cite{wolsey2020integer}, 
and hence potentially more challenging to solve than the LP or QP required for linear MPC. 

The MIQPs can in MPC either be solved in real-time online or be formulated as a multi-parametric MIQP (mp-MIQP) and solved offline parametrically for a range of states. Once the problems have been solved offline, the necessary computations online are reduced to evaluating a look-up table, \cite{bemporad2002explicit}. 
However, it is very challenging to solve and compute an efficient data structure for a parametric solution of an mp-MIQP. Furthermore, a potentially large storage for the look-up table of the solution in the embedded system is required. For these reasons, the online approach is often the only realistic solution. For that approach to also be considered reliable for more critical applications, a priori guarantees for that the computational requirements of the problem at hand do not exceed the hardware capabilities are desirable. Similar work for LPs and QPs can be found in \cite{zeilinger2011real} and \cite{cimini2017exact,cimini2019complexity,arnstrom2021unifying}, respectively. The first work in that direction for MIQPs can be found in \cite{axehill2010improved}. 
The strategy to achieve this, and to obtain a significantly more practically useful complexity analysis than the conservative classical one, is to consider a specific set of MIQP problems encoded as an mp-MIQP. Each MIQP problem to be solved online is an instance of the mp-MIQP for a fixed parameter. Furthermore, beyond the traditional algebraic theoretical analysis, a certification algorithm is developed and used offline to analyze the complexity of all problems that might be requested to be solved online.

In this work, the optimization problem online is to be solved using a branch-and-bound-based MIQP solver. To compute the solution of an MIQP with integer variables modeled as binaries, a straightforward approach is to enumerate all possible combinations of binary variables and solve a QP for each such combination. As the number of binary variables increases, however, the computational effort for this approach grows exponentially. An alternative solution strategy commonly offering a remedy to this problem is the branch and bound (\bnb) method \cite{zhang1996branch,axehill2010improved,almer2013efficient,axehill2014parametric}, 
in which a sequence of relaxed QP problems are ordered and solved in a binary search tree to find an optimal mixed-integer solution. In the worst case, this method still requires the solution of exponentially many QPs, however, the complexity observed in practice is often significantly lower than explicit enumeration. The objective with the work in this paper is to continue the work in \cite{axehill2010improved} 
where useful sufficiently exact complexity guarantees are computed a priori for branch-and-bound focusing on the size of the binary search tree. The main contribution in this work is to also analyze the complexity of the QP problems for the relaxations in the nodes of the binary search tree. This is done using recently presented state-of-the-art methods for exact complexity certification for active-set QP solvers, \cite{cimini2017exact,arnstrom2021unifying}, 
where the exact bound on the worst-case number of iterations required by an active-set QP solver is determined. 
Hence, the focus in this work is on the overall
\emph{worst-case complexity} in a \bnb-based MIQP solver, in terms of the worst-case number of QP iterations in the (worst-case sequence of) subproblems necessary to solve in order to compute the optimal mixed-integer solution. This knowledge enables us to compute relevant complexity bounds on the worst-case complexity of MIQP solvers based on \bnb, which is of significant importance in the context of real-time MPC for hybrid systems.

The organization of the paper is as follows. Section  \ref{sec2} introduces the mp-MIQP problem formulation. Some background theory and the \bnb algorithm for solving MIQPs are revisited in Section \ref{sec3}. The main contribution of the paper, i.e., an algorithm for certification of MIQP, is presented in Section  \ref{sec:cert}.  
Finally, numerical experiments are provided in Section \ref{sec:num-ex} to illustrate the result of the proposed algorithm.


        	\section{Problem Formulation}  \label{sec2}
	For an MLD system, by considering the state variables as parameter vector $\theta$ and control actions as optimization variables $x$, the hybrid MPC problem can be cast into an mp-MIQP problem as follows \cite{bemporad1999control,axehill2010improved},  
	\begin{subequations} \label{eq3}
		\begin{align}
		\vspace{.1cm} 
		\min_{x} \quad & \frac{1}{2} x^{T}Hx + f^T x + \theta ^{T} f_{\theta }^{T}x, \label{eq3_1} \\		
		\textrm{s.t.} \quad & Ax \leqslant  b + W \theta, \label{eq3_2} \\
		& x_i \in  \{0,1\}, \hspace{.3cm} \forall i \in \mathcal{B} \label{eq3_3}
		\end{align}
	\end{subequations}
	where the decision variable consists of $n_c$ continuous and $n_b$ binary variables, i.e., $x = [x_c^T, x_b^T]^T \in \mathbb{R}^{n_c} \times \{0, 1\}^{n_b}$, $n = n_c + n_b$, and the parameter vector is $\theta \in \Theta_0 \subset \mathbb{R}^{n_{\theta}}$. The set $\Theta_0$ is assumed to be polyhedral. The MIQP problem, denoted by $\mathcal{P}(\theta)$, is given by
	$H \in \mathbb{S}_{++}^n$, $f \in \mathbb{R}^n$, $f_{\theta} \in \mathbb{R}^{n \times n_{\theta}}$, and the feasible set is determined by $A \in \mathbb{R}^{m \times n}$, $b \in \mathbb{R}^{m}$, and $W \in \mathbb{R}^{m \times n_{\theta}}$. 	
	Since the $n_b$ optimization variables indexed by the set $\mathcal{B}$ are binary-valued, the problem is no longer convex and is known to be $\mathcal{NP}$-hard \cite{wolsey2020integer}. 
	
	Throughout this paper, the notation $\{\mathcal{U}^i\}_{i=1}^N$ denotes a finite collection $\{\mathcal{U}^1, \dots \mathcal{U}^N\}$ of $N$ elements. When $N$ is unimportant, we use the notation $\{\mathcal{U}^i\}_i$ instead. 

        \section{Optimization preliminaries} \label{sec3}
In this section, some useful optimization preliminaries to be able to understand \bnb is briefly discussed.
\subsection{Quadratic Programming}
\label{sec:QP}
 The QP problem considered in this work is given by, 
\begin{subequations} \label{eqQP}
	\begin{align}
	\vspace{.1cm} 
	\min_{x} \quad & \frac{1}{2} x^{T}Hx + f^T x, \label{eqQP_1} \\	 
	\textrm{s.t.} \quad & Ax \leqslant  b, \label{eqQP_2}\\	
	 \quad & A_{\mathcal{E}}x =  b_{\mathcal{E}} \label{eqQP_3}
	\end{align}
\end{subequations}	
where $x \in \mathbb{R}^n$,  $f \in \mathbb{R}^n$, $b \in \mathbb{R}^{m}$,  $b_{\mathcal{E}} \in \mathbb{R}^{p}$ and matrices $H \in \mathbb{S}_{++}^n$, $A \in \mathbb{R}^{m \times n}$, and $A_{\mathcal{E}} \in \mathbb{R}^{p \times n}$. 
There are several methods to compute the solution of \eqref{eqQP}, such as active-set methods, interior point methods, and various gradient methods, \cite{nocedal2006numerical}. 
In this paper the (dual) active-set method in \cite{arnstrom2021efficient} 
is applied, since these are known to be very suitable in branch-and-bound. Even though this work does not currently consider that aspect, the active-set methods' warm starting capabilities are of great interest in branch and bound. Another important reason for this choice of method is that there exist recently presented certification methods for this solver, \cite{cimini2017exact,arnstrom2021unifying}. 
	  
Active constraints is an important concept in this type of methods which defines whether a constraint holds with equality, e.g., constraint $i$ in \eqref{eq3_2} is active if $A_i x(\theta) =  b_i + W_i \theta$, where subscript $i$ denotes the $i$th row of a matrix or vector. The optimal active set, denoted by $\mathcal{A}$, is defined as the set containing the indices of all constraints active at the optimal solution.
	
The idea of the active-set method is to iteratively make steps towards the optimal solution by solving a sequence of equality constrained QPs, in which some of the inequality constraints are imposed as active. These constraints are said to be included in the working set. See \cite{nocedal2006numerical} for a detailed description of the active-set method. 
In this work, the number of equality constrained QPs that the active-set solver solves to find the optimal solution of the QPs will be called the iteration number. In \cite{arnstrom2021unifying}, it is shown how this number can be found a priori and also how the relevant sequences of working-set changes can be computed as a function of the parameters in an mp-QP.

\subsection{Mixed-Integer Quadratic Programming} \label{subsec:MIQP}
 Consider the following MIQP problem,
\begin{subequations} \label{eq5}
	\begin{align}
	\vspace{.1cm} 
	\min_{x} \quad & \frac{1}{2} x^{T}Hx + \bar{f}^T x, \label{eq5_1} \\		
	\textrm{s.t.} \quad & Ax \leqslant  \bar{b}, \label{eq5_2} \\
	& x_i \in  \{0,1\}, \hspace{.3cm} \forall i \in \mathcal{B} \label{eq5_3}
	\end{align}
\end{subequations}	
Each instance obtained from the original problem \eqref{eq3} by fixing the parameters to a specific value $\bar{\theta}$ will be a problem in the form in~\eqref{eq5}, with $\bar{f} = f + f_{\theta }  \bar{\theta}$ and $\bar{b} = b + W \bar{\theta}$. Therefore, we denote the problem $\mathcal{P}(\bar{\theta})$. The decision variables, $H, A$, and $\mathcal{B}$ in \eqref{eq5} are defined as in \eqref{eq3}. 

In this work, the branch and bound (\bnb) method is used to solve the optimization problem in~\eqref{eq5}. In the \bnb search tree, the binary constraints are relaxed into interval constraints 
forming a so-called relaxation given by,
\begin{subequations} \label{eq4}
			\begin{align}
			\min_{x} \quad & \frac{1}{2} x^{T}Hx + \bar{f}^T x, \label{eq4_1} \\		
			\textrm{s.t.} \quad & Ax \leqslant  \bar{b}, \label{eq4_2} \\
			& 0 \leqslant x_i  \leqslant 1, \hspace{.3cm} \forall i \in \mathcal{B}, \label{eq4_3}\\
			& x_i = 0, \hspace{.1cm} \forall i \in \mathcal{B}_0, \hspace{.3cm} x_i = 1, \hspace{.1cm} \forall i \in \mathcal{B}_1 \label{eq4_4}
			\end{align}
		\end{subequations}
that is a convex QP problem in the form in \eqref{eqQP}, where $\mathcal{B}_0,\mathcal{B}_1 \subseteq \mathcal{B}$ and $\mathcal{B}_0 \cap \mathcal{B}_1 = \emptyset$. We denote this problem  $\mathcal{P}(\bar{\theta},\mathcal{B}_0,\mathcal{B}_1)$. 

 The procedure of \bnb is as follows. Starting in the root node, all binary constraints \eqref{eq5_3} in $\mathcal{P}(\bar{\theta})$ are relaxed to \eqref{eq4_3} resulting in a fully relaxed QP problem 
which is solved. At the next level in the tree, one of the remaining binary-constrained variables is fixed to $0$ and $1$, resulting into two new subproblems that form two new nodes in the \bnb tree which are called the children of the parent node. The procedure is repeated and new nodes are explored further down in the tree. In the bottom of the tree the leaf nodes are found in which all binary constraints have been relaxed.

An important property of \bnb is that the relaxation gives a lower bound on the optimal integer solution for the subtree below the node under consideration, whereas an integer feasible solution provides an upper bound on the value function valid in the entire tree. To use the result from the relaxations to prune parts of the search tree from explicitly being explored is a fundamental idea in \bnb. When the solution to a relaxation in a node is shown to satisfy one of the following cut conditions, then all nodes in the subtree below that node can be shown to be of no use and can therefore be disregarded, \ie, the tree is cut in that node \cite{almer2013efficient}: 
\begin{enumerate} 
	\item The relaxation is \emph{infeasible}. Adding constraints will not change that. Hence, the entire subtree below that node is infeasible.
	\item The optimal objective function value of the relaxation is \emph{greater} than the one of the best known integer solution so far. Adding constraints will not decrease the objective function value. Hence, a better objective function value cannot be found in the subtree below that node.
	\item The solution to the relaxation is \emph{integer feasible}. Adding constraints will not decrease the objective function value. Hence, an optimal solution for the entire subtree below that node has already been found.
\end{enumerate}
In this work, we assign the value function of an infeasible relaxed problem to infinity. Therefore, the first and second cut conditions can be tested simultaneously. In comparing the value function with the upper bound, if the comparison holds, that is if the node has an objective function value that is worse than the best known integer solution, the node will be cut. The integer feasibility cut condition is tested  
by looking at the active set, such that if the binary constraints \eqref{eq5_3} are active, i.e., all the binary decision variables are either $0$ or $1$, then the solution is integer feasible and the tree is pruned and the upper bound is updated. If none of the cut conditions hold, the node is split into two new subproblems by fixing a relaxed binary variable indexed by $k$ to $0$ and $1$, forming two new subproblems $\mathcal{P}(\bar{\theta},\mathcal{B}_0 \cup \{k\},\mathcal{B}_1)$ and $\mathcal{P}(\bar{\theta},\mathcal{B}_0,\mathcal{B}_1 \cup \{k\})$, respectively. These two subproblems are inserted in the sorted list $\mathcal{T}$, implementing a priority queue, to be analyzed as the tree is explored. The priority used in~$\mathcal{T}$ is determined by the choice of the tree exploration strategy (e.g., depth-first) and whether the left or right branch should be explored first

Algorithm \ref{alg1} summarizes the \bnb method for solving problem \eqref{eq5}. In this algorithm, the sorted list $\mathcal{T}$ stores the subproblems to be solved in a user-defined exploration order, $\bar{x}$ is the best integer feasible solution so far, and $\bar{J}$ is its corresponding objective function value (the upper bound).  
Moreover, $x$ and $J$ are the associated optimal solution and objective function value of the relaxation, respectively. As mentioned earlier, if the solution of a subproblem turns out to be infeasible, we set $J = \infty$. 

The QP relaxation  $\mathcal{P}(\bar{\theta},\mathcal{B}_0,\mathcal{B}_1)$ \eqref{eq4} is solved using an active-set method as outlined in Section~\ref{sec:QP}.
Of particular interest in this work is the number of QP iterations required to solve a relaxation, and the sum of all such QP iterations over the entire \bnb tree for each choice of the parameter.	


\begin{algorithm}[H]
	\caption{Branch and Bound for MIQP}
	\label{alg1}
	\begin{algorithmic}[1]
		\Require MIQP problem $\mathcal{P}(\bar{\theta})$ for $\bar{\theta}$ given 
		\Ensure $\bar{J}, \bar{x}$ 
		\State{$ \bar{J} \leftarrow \infty , \bar{x} \leftarrow $ void}
		\State{Push $\mathcal{P}(\bar{\theta},\emptyset,\emptyset)$ to $\mathcal{T}$}
		\While {$\mathcal{T} \neq \emptyset$}
		\State Pop $\mathcal{P}(\bar{\theta},\mathcal{B}_0,\mathcal{B}_1) $ from  $\mathcal{T}$ 
		\State  $J, x \leftarrow$ Solve $\mathcal{P}(\bar{\theta},\mathcal{B}_0,\mathcal{B}_1) $ \label{alg1_solv} 
		\If {$J \geq \bar{J} $} \label{alg1-cut2}
		\State There exists no feasible solution to $\mathcal{P}(\bar{\theta},\mathcal{B}_0,\mathcal{B}_1)$  which is better than $\bar{x} $ 
		\ElsIf {all binary variables are active } \label{alg1-cut3}		
		\State  Better integer feasible solution is found 
		\State $\bar{J} \leftarrow J $
		\State $\bar{x} \leftarrow x $
		\Else
		\State Select $k$ : $k \in \mathcal{B}$, $k \notin (\mathcal{B}_0 \cup \mathcal{B}_1) $ 
		\State Push $\mathcal{P}(\bar{\theta},\mathcal{B}_0 \cup \{k\},\mathcal{B}_1)$ and  $\mathcal{P}(\bar{\theta},\mathcal{B}_0,\mathcal{B}_1 \cup \{k\})$ to $\mathcal{T}$ \label{alg1-push}
		\EndIf		
		\EndWhile
	\end{algorithmic}
\end{algorithm}
%
%


			\section{Complexity certification of MIQP} \label{sec:cert}
	In this section, an algorithm that analyzes Algorithm \ref{alg1} for all values of the parameter $\theta$ is presented. As mentioned, in hybrid MPC, these parameters could be system states or\slash and reference signals. 
	To certify the solution process of an MIQP, the iterations of Algorithm \ref{alg1} are executed parametrically with respect to $\theta$, for a given parameter set $\Theta_0$. In particular, the framework to be presented is able to rigorously analyze and upper bound the computational complexity for the process, in terms of size of the \bnb tree or total number of QP iterations. 	
	
	Algorithm \ref{alg2} represents the proposed parametric \bnb-based MIQP certification algorithm with the aim of certifying the solution process of the problem in \eqref{eq3} in detail. An important property of the algorithm is that it iteratively divides and explores the parameter space based on the polyhedral partition from a QP certification algorithm. Here, $\kappa(\theta)$ denotes a complexity measure for solving the QP relaxation \eqref{eq4} denoted by $\mathcal{P}(\theta,\mathcal{B}_0,\mathcal{B}_1)$. In addition to the usual sorted list $\mathcal{T}$ in \bnb, there are two additional lists of tuples in Algorithm \ref{alg2}. One is the list $\mathcal{S}$ holding regions in the parameter space where the analysis has not completed yet. It contains tuples $\left ( \Theta,\kappa,\mathcal{T},\bar{\mathbb{J}}(\theta) \right )$, where $\Theta$ is the associated polyhedral parameter set, $\mathcal{T}$ is the normal list (\cmp Algorithm~\ref{alg1}) of \bnb to store subproblems, $\kappa$ is the complexity measure for $\theta \in \Theta$,  
	 and $\bar{\mathbb{J}}(\theta)$ is a collection of the upper bounds on the value function. Another related list is $\mathcal{F}$ which contains regions in the parameter space where the process has terminated, \ie, $\mathcal{F}$ holds the final partition. It consists of a tuple with the information $(\Theta,\kappa)$ for every terminated region with the associated complexity measure.	 
	 	
	For what follows, the following definition is required.
	\begin{definition} \label{def1}	
		A set of polyhedra $\{\mathcal{R}^i\}_{i=1}^N$ is a \emph{polyhedral partition} of $\Theta$ if $ \mathring{\mathcal{R}}^i \cap \mathring{\mathcal{R}}^j = \emptyset, i\neq j$, and $\cup_{i=1}^{N} \mathcal{R}^i = \Theta$, where  $\mathring{\mathcal{R}}^i$ denotes the interior of the region $\mathcal{R}^i$.
	\end{definition} 	

	At each iteration of Algorithm \ref{alg2}, a region $\Theta$ will be popped from $\mathcal{S}$. If the list $\mathcal{T}$ associated to that region is empty, the exploration of the tree has been completed for all parameters in $\Theta$ and is, hence, added to the final partition $\mathcal{F}$ (Step \ref{step:add-to-final}). Otherwise a node in $\mathcal{T}$ will be selected and the corresponding mp-QP relaxation will be solved and certified over $\Theta$ in Step \ref{alg2_cert}. 
	The function \textsc{QPcert} 
	here is the complexity certification algorithm for QPs which takes an mp-QP problem $\mathcal{P}(\theta,\mathcal{B}_0,\mathcal{B}_1)$ and a parameter region $\Theta$ as inputs and returns a polyhedral partition $\{\Theta^j\}_{j=1}^{N}$ of $\Theta$. In each region, the function provides the computational complexity $\kappa^j$ for solving $\mathcal{P}(\theta,\mathcal{B}_0,\mathcal{B}_1)$ for all parameters in $\Theta^j$, as well as the explicit solution of the mp-QP relaxation including the optimal active set $\mathcal{A}^j$ and the value function $J^j(\theta)$ which is convex piecewise quadratic over the polyhedral region. 
	If $\mathcal{P}(\theta)$ is infeasible in $\Theta^j$, we encode that by $J^j(\theta)=\infty$. 
	
	After the mp-QP relaxation has been processed, the sorted list $\mathcal{T}$ is updated in the  \textsc{updateTree} procedure for each new region based on the optimal active set $\mathcal{A}^i$ and the value function $J^{i}$. The update consists of testing the three cut conditions described in Section \ref{subsec:MIQP} \textit{parametrically}. If none of these conditions are invoked for $\Theta^i$, two new nodes are created according to standard \bnb (based on $\mathcal{A}^i$ and the binaries that are fixed in the current node) and added to $\mathcal{T}^i$. 		
		
	The \bnb algorithms in this work satisfy the following assumptions. 
	\begin{assumption} \label{assumption-order-tree}
		The tree exploration strategy is depth-first in Algorithms \ref{alg1} and \ref{alg2}. Moreover, the branch variable order is fixed
		(e.g., not parameter dependent) and known beforehand, and the order in which the $0$ and $1$ branch will be explored is fixed and known beforehand in both algorithms.
	\end{assumption}	
	
	 The proposed certification method is formally analyzed in the following subsections as follows. The spatial decomposition employed in Algorithm \ref{alg2} is described in Section \ref{subsec:spatial-decom}. Properties of the parametric \bnb search tree are formalized in Section \ref{subsec:prop of tree} and finally, the main contributions of the proposed algorithm are stated in Theorems \ref{theorem-sequence}  and \ref{cor-conserve-cert} in Section \ref{subsec:prop of cert}.
	 
			\begin{algorithm}[H]
				\caption{Certification of Branch and Bound for MIQP} 
				\label{alg2}
				\begin{algorithmic}[1] 
					\Require mp-MIQP problem $\mathcal{P}(\theta)$, and $\Theta_0$
					\Ensure Partition $\mathcal{F}=\{(\Theta^i,\kappa^i)\}_{i=1}^{N_f}$
					\State $\kappa_0 \leftarrow 0$, $\bar{\mathbb{J}}_0 \leftarrow \{\infty\}$ 
					\State $\mathcal{F} \leftarrow \emptyset$ 
					\State{Push $\mathcal{P} \left (\theta, \emptyset, \emptyset \right )$ to $\mathcal{T}_0$}
					\State Push $(\Theta_0, \kappa_0, \mathcal{T}_0,\bar{\mathbb{J}})$ to $\mathcal{S} $ 
					\While {$\mathcal{S} \neq \emptyset$} \label{step:outer-loop}
					\State Pop $ \left ( \Theta,\kappa, \mathcal{T},\bar{\mathbb{J}} \right )$ from $\mathcal{S} $ 
					\If {$\mathcal{T}=\emptyset$}
					\State Push $(\Theta,\kappa)$ to $\mathcal{F}$  and \textbf{goto} Step \ref{step:outer-loop} 
					\label{step:add-to-final}
					\EndIf
					\State Pop $\mathcal{P} \left ( \theta,\mathcal{B}_0,\mathcal{B}_1 \right )$ from  $\mathcal{T}$ \label{alg2_pop}
					\State $\{(\Theta^i, \kappa^i, \mathcal{A}^i, J^{i})\}_{i=1}^N \leftarrow \textsc{QPcert}( \mathcal{P}(\theta,\mathcal{B}_0,\mathcal{B}_1),\Theta )$  
					\label{alg2_cert} 
					\For {$j\in\{1,\dots,N\}$} \label{alg2_for}
					\State $\mathcal{T}^i, \bar{\mathbb{J}}^i \leftarrow \textsc{updateTree}(\mathcal{T},\Theta^i, \mathcal{A}^i, J^i,\bar{\mathbb{J}})$ 
					\State Push $ \left (\Theta^i,\kappa+\kappa^i,\mathcal{T}^i,\bar{\mathbb{J}}^i \right )$ to $\mathcal{S}$ \label{alg2_increase} 
					\EndFor
					\EndWhile
					\\ \hrulefill
					\Procedure{updateTree}{$\mathcal{T}^j,\Theta^j, \mathcal{A}^j, J^j,\bar{\mathbb{J}}^j$}
					\If {$ \exists \bar{J} \in \bar{\mathbb{J}}^j$ : $ J^{j}(\theta) \geq \bar{J}(\theta)$, $\forall \theta\in \Theta^j$} 
					\label{alg2_comp}	
					\State There exists no $\theta$ for which a feasible solution of $\mathcal{P}$  provides better solution 
					\ElsIf {all binary variables are in $\mathcal{A}^j$}  \label{alg2_feas}
					\State  Potentially better integer feasible solution has been found 
					\State $\bar{\mathbb{J}}^j(\theta) \leftarrow \bar{\mathbb{J}}^j(\theta) \cup \{J^j(\theta)\}$
					\Else				
					\State Select $k$ : $k \in \mathcal{B}$, $k \notin (\mathcal{B}_0 \cup \mathcal{B}_1) $ 
					\State Push $\mathcal{P}(\theta,\mathcal{B}_0 \cup \{k\},\mathcal{B}_1)$ and $\mathcal{P}(\theta,\mathcal{B}_0,\mathcal{B}_1 \cup \{k\})$ to $\mathcal{T}^j$ 
					\EndIf		
					\State \textbf{return} $\mathcal{T}^j, \bar{\mathbb{J}}^j$ 
					\EndProcedure 
				\end{algorithmic}
			\end{algorithm}
				
	\subsection{Decomposition of the parameter space} \label{subsec:spatial-decom} 
	As mentioned above, an interpretation of Algorithm \ref{alg2} is that it performs iterations of Algorithm \ref{alg1} parametrically. As a result, the parameter space will be partitioned depending on, for example, how the exploration of the \bnb tree differs for different parameters. We clarify how this partitioning occurs and some of its properties in this subsection.  

	At Step \ref{alg2_cert}, after processing the node, it is divided spatially into $N$ subproblems with corresponding regions $\Theta^i$, each holding a copy of the \bnb-list $\mathcal{T}$. The \textsc{updateTree} procedure which tests the pruning conditions parametrically is then applied for each local tree.
	
	The infeasibility and integer feasibility cut conditions can directly be applied to the output from $\textsc{QPcert} \left (\mathcal{P}(\theta,\mathcal{B}_0,\mathcal{B}_1),\Theta \right)$, given by $J^i(\theta)=\infty$ and the active set $\mathcal{A}^i$ for the region, respectively, with no further spatial partitioning. In the dominance cut condition however, if it is done exactly, further spatial partitioning is imposed by $ J^{i}(\theta) \geq \bar{J}(\theta)$ additionally splitting the node spatially into two new subproblems with finer regions. 
	To make the certification more tractable by maintaining the polyhedral structure of the partition, at the price of conservatism, the spatial partitioning originating from the dominance cut can be avoided by only cutting if the condition is satisfied for the entire region, \ie, $ J^{j}(\theta) \geq \bar{J}(\theta)$, $\forall \theta\in \Theta^j$, similarly to what was done in~\cite{axehill2010improved}.
	Note that the framework is in principle capable of exact certification, however, the price for that is that the geometry becomes non-polyhedral and will in general include regions described by quadratic functions. To investigate the details of the less tractable exact case is therefore decided to be beyond the scope of this work. The price for not considering the exact case is that some conservatism will be included in the certification result and multiple upper bounds related to what in parametric programming for MIQPs is known as ``overlaps'' have to be introduced. Note that, however, the result from Algorithm \ref{alg2} is guaranteed to \emph{upper} bound the online complexity.	
	The end result of the proposed algorithm is many local \bnb trees, which represent all possible \bnb trees that are used to compute the solution of the problem for all parameter values online. 		

 	 In the following, we ensure that the partitioning made in Algorithm 2 is correct, i.e., the generated regions cover the entire $\Theta_0$ and have no overlap among themselves. 
 	       
	\begin{lemma}[Maintenance of complete partition]
		\label{lem:maintenance}
		At an arbitrary iteration in Algorithm \ref{alg2}, the union of regions in $\mathcal{F}$ and $\mathcal{S}$ forms a partition of $\Theta_0$. 
	\end{lemma}
	\begin{proof}
		We will prove the lemma by induction. Let $\{\Theta^i_S\}_i$ and $\{\Theta^i_{\mathcal{F}}\}_i$ denote the regions contained in $\mathcal{S}$ and $\mathcal{F}$, respectively, at the start of an iteration. Similarly, let $\{\Theta^i_{S'}\}_i$ and $\{\Theta^i_{\mathcal{F}'}\}_i$ denote the regions contained in $\mathcal{S}$ and $\mathcal{F}$, respectively, at the end of the iteration. Finally, let $\mathcal{T}$ and $\Theta$ denote the list and region, respectively, that are selected from $\mathcal{S}$ at the start of the iteration.

		If $\mathcal{T}=\emptyset$ we get $\{\Theta_{\mathcal{F}'}^i\}_i = \{\Theta^i_{\mathcal{F}}\}_i \cup \{\Theta\}$ and $\{\Theta^i_{S'}\}_i = \{\Theta_S^i\}_i\setminus \{\Theta\}$ after Step~\ref{step:add-to-final} is executed. Clearly $\{\Theta^i_{\mathcal{F'}}\}_i$ and $\{\Theta^i_{S'}\}_i$ will form a partition of $\Theta_0$ if $\{\Theta^i_{\mathcal{F}}\}_i$ and $\{\Theta^i_S\}_i$ do, since we have the same underlying regions (only a single region $\Theta$ has been moved from $\mathcal{S}$ to $\mathcal{F}$ during the iteration). 
		
		Next, consider the case when $\mathcal{T}\neq \emptyset$. Then the regions in $\mathcal{F}$ remain unchanged and we will partition $\Theta$ into $\{\Theta^i\}_i$ in Step~\ref{alg2_cert}. Hence, we get $\{\Theta^i_{S'}\}_i = \left(\{\Theta^i_{S}\}_i \setminus\{\Theta\}\right) \cup \{\Theta^i\}_i$.
		Then, since $\cup_i\Theta^i = \Theta$ from the properties of \textsc{QPcert}, it follows that $\cup_i \Theta^i_{S'} = \cup_i \Theta^i_S$. Moreover, since $\Theta^i \subseteq \Theta$ $\forall i$ and that $\mathring{\Theta}^i \cap \mathring{\Theta}^j = \emptyset$ if $i\neq j$ (both, again, from the properties of \textsc{QPcert}), all regions in $\{\Theta^i_{S'}\}_i$ and $\{\Theta^i_{\mathcal{F}'}\}_i$ are disjoint if the regions $\{\Theta^i_S\}_i$ and $\{\Theta^i_{\mathcal{F}}\}_i$ are. In conclusion, the regions in $\mathcal{S}$ and $\mathcal{F}$ at the end of an iteration will form a partition of $\Theta_0$ if the regions in $\mathcal{S}$ and $\mathcal{F}$ do so at the start of the iteration (completing the induction step). 
		
		For the base case we have that $\mathcal{F} = \emptyset$ and $\mathcal{S}$ only contains $\Theta_0$ at the start of the algorithm, which trivially forms a partition of $\Theta_0$. 
	\end{proof}
	As a special case, Lemma 1 ensures that the entire region of interest $\Theta_0$ will be considered before termination.
	\begin{corollary}[Complete partition at termination]
		Assume that Algorithm \ref{alg2} has terminated with the partition $\{(\Theta^i,\kappa^i)\}_i$.
		Then $\cup_i\Theta^i = \Theta_0$ and $\mathring{\Theta}^i \cap \mathring{\Theta}^j = \emptyset$ for $i\neq j$. That is, $\Theta_i$ is a partition of $\Theta_0$.
	\end{corollary}
	\begin{proof}
		Follows from Lemma \ref{lem:maintenance} and that Algorithm \ref{alg2} terminates when $\mathcal{S}=\emptyset$.
	\end{proof}
	
%
%
	\subsection{Properties of the search tree}	\label{subsec:prop of tree}
	As clarified in Section \ref{subsec:spatial-decom}, the parametric analysis of the dominance cut is relaxed in this work since the cost of an exact analysis would be quadratic regions. To ensure that Algorithm \ref{alg2} still provides correct complexity bounds, despite this relaxation, we review some fundamental properties of the B\&B tree below, which we then use in Section \ref{subsec:prop of cert} to prove the correctness of the proposed method.
	
	Let a \textit{node} be characterized by a tuple $\eta \triangleq (\mathcal{B}_0, \mathcal{B}_1)$, where $\mathcal{B}_0$ and $\mathcal{B}_1$ are defined in \eqref{eq4}. The following definition is required to formalize the properties of the parametric \bnb. 
	
	\begin{definition} \label{descendent}	
	 The node $\eta=(\mathcal{B}_0, \mathcal{B}_1)$ is a \textit{descendant} to node $\hat{\eta}=(\hat{\mathcal{B}}_0, \hat{\mathcal{B}}_1)$, denoted by  $\eta\in\mathcal{D}(\mathcal{\hat{\eta}})$, if $\mathcal{B}_0 \supseteq \hat{\mathcal{B}}_0$ and $\mathcal{B}_1 \supseteq \hat{\mathcal{B}}_1$. 
	\end{definition}
	\begin{lemma}
		\label{lem:dominance-parent}
		Let $J^*_{\eta}(\theta)$ denote the value function for the mp-QP relaxation in node $\eta$ and let $\delta \in \mathcal{D}(\eta)$, then $J^*_{\eta}(\theta)\leq J_{\delta}^*(\theta)$. 
	\end{lemma}
	\begin{proof}
		It directly follows from standard arguments from \bnb that the feasible set for the mp-QP relaxation in node $\delta$ is a subset of the feasible set for the mp-QP relaxation in node $\eta$, since some additional equality constraints have been added in $\delta$. 
	\end{proof}
	Let $\mathbb{B}_\eta(\theta)$ be the set of nodes which have yielded an integer feasible solution up until node $\eta$ is processed, given an ordering of $\mathcal{T}$. Moreover, let $\bar{J}_{\eta}(\theta)$ denote the lowest (\ie, best) upper bound found up until $\eta$ is processed, i.e., 
	\begin{equation}
	\bar{J}_{\eta}(\theta) \triangleq \min_{\tilde{\eta}\in\mathbb{B}_{\eta}(\theta)} J^*_{\tilde{\eta}}(\theta). 
	\end{equation}
	
	The following lemma shows that if the dominance cut condition is invoked for node $\eta$, the best upper bound $\bar{J}_{\eta}(\theta)$ will not be changed (in particular not improved) by processing any descendant of $\eta$. 
	\begin{lemma}
		\label{lem:redundant-desc}
		Assume that $J^*_{\eta}(\theta)\geq \bar{J}_{\eta}(\theta)$ in a node $\eta$. Then the dominance cut condition will be satisfied in that node and 
		\begin{equation}
		\min_{\tilde{\eta}\in \mathbb{B}_{\eta}(\theta)\cup \mathcal{D}(\eta)} J^*_{\tilde{\eta}}(\theta) = \bar{J}_{\eta}(\theta).
		\end{equation}
	\end{lemma}
	\begin{proof}
		We have that  
		\begin{equation*}
		\begin{aligned}
		\min_{\tilde{\eta}\in \mathbb{B}_{\eta}(\theta)\cup \mathcal{D}(\eta)} J^*_{\tilde{\eta}}(\theta) 
		&= \min\left\{\min_{\tilde{\eta}\in\mathbb{B}_{\eta}(\theta)}J^*_{\tilde{\eta}}(\theta),\min_{\delta\in \mathcal{D}(\eta)} J_{\delta}^*(\theta) \right\} \\
		&= \min\left\{\bar{J}_{\eta}(\theta), \min_{\delta\in \mathcal{D}(\eta)} J_{\delta}^*(\theta) 
		\right\} \\
		&= \bar{J}_{\eta}(\theta),
		\end{aligned}
		\end{equation*}
		where the last equality follows from $J_{\delta}^*(\theta)\geq J_{\eta}^*(\theta) \geq \bar{J}_{\eta}(\theta)$ $\forall\delta\in \mathcal{D}(\eta)$, i.e., from Lemma \ref{lem:dominance-parent} and the premise. 
	\end{proof}
	Hence, processing descendants of a node for which the dominance cut was invoked would not change the best known upper bound. As a result, the exploration of, and application of cut conditions in, the rest of the tree will not be affected if $\mathcal{T}\cup D(\eta)$ is used instead of $\mathcal{T}$ after the dominance cut is invoked in node $\mathcal{\eta}$. We formalize this notion in the following corollary.
	
	\begin{corollary}[Futility of processing dominated descendants]
		\label{cor:redundant-desc-exploration}
		Consider two cases of the list $\mathcal{T}$ after the dominance cut is invoked when processing node $\eta$,
		\begin{enumerate}
			\item $\mathcal{T}$ is used (normal operation) 
			\item $\mathcal{T}\leftarrow \mathcal{T}\cup \mathcal{D}(\eta)$ is used (redundant addition) 
		\end{enumerate}
		Let $\mathbb{B}^1$ be the set of \textit{all} nodes explored during case 1 and let $\mathbb{B}^2$ be the set of \textit{all} nodes explored during case 2. Then $\mathbb{B}^2 \setminus \mathbb{B}^1 = \mathcal{D}(\eta)$. That is, exploring the nodes in $\mathcal{D}(\eta)$ does not affect the subsequent exploration of nodes in the tree. 		
	\end{corollary}
	\begin{proof}
		The only thing that might alter the exploration is if the upper bound is changed by processing additional nodes (which could prune nodes that would otherwise be explored). However, from Lemma \ref{lem:redundant-desc} processing any descendants of $\eta$ can never change the best upper bound, which in turn means that further exploration is unaffected by processing any $\delta \in \mathcal{D}(\eta)$.
	\end{proof}
	%
\subsection{Properties of the certification}	\label{subsec:prop of cert}	
	So far we have shown the parametric behavior of Algorithm \ref{alg2} and the properties of the search tree. 
	In this subsection, the properties of the proposed certification method are derived in the following theorems.
	\begin{theorem}	 \label{theorem-sequence} 
		Assume Assumption \ref{assumption-order-tree} holds. 
		Let $\mathbb{B}(\theta)$ denote the set of all nodes explored to solve the problem in \eqref{eq5} for a fixed $\theta$ using Algorithm~\ref{alg1}. Moreover, let $\hat{\mathbb{B}}_k(\theta)$ denote all nodes explored in Algorithm \ref{alg2} for region $\Theta_k$ before terminating. Then $\mathbb{B}(\theta) \subseteq \hat{\mathbb{B}}_k(\theta)$, $\forall \theta\in \Theta_k$.
	\end{theorem}
	\begin{proof} \label{proof-theorem-sequence}  
		First, for the analysis to be meaningful, we stress that it is implicitly assumed that \textsc{QPcert} certifies the QP solver correctly. In particular, this means that the active set and feasibility is correctly identified for each parameter value. As a result, the correctly identified active set and infeasibility directly leads to correctly invoked integer feasibility and infeasibility cuts, respectively. What remains is the dominance cut condition, which will now be analyzed in more detail.
		Consider the processing of node $\eta$ on region $\Theta$ in Algorithm \ref{alg2} and assume that a conservative upper bound $\tilde{J}(\theta)$ on $\Theta$ is used for the dominance cut (i.e., $\tilde{J}(\theta)\geq \bar{J}_{\eta}(\theta), \forall \theta\in \Theta$). Then we have the following cases than could happen in Step \ref{alg2_comp} of Algorithm \ref{alg2}:
		\begin{enumerate}
			\item $J^*_{\eta}(\theta) \geq \tilde{J}(\theta)$ $\forall \theta\in \Theta$ $\rightarrow$ dominance cut invoked correctly $\forall\theta \in \Theta$. 
			\item $\exists \tilde{\theta}\in \Theta$ such that $ J^*_{\eta}(\tilde{\theta}) <   \tilde{J}(\theta)$ 
			\begin{enumerate}
				\item $J^*_{\eta}(\theta) < \bar{J}_{\eta}(\theta)$ $\forall \theta\in \Theta$ $\rightarrow$ dominance cut dismissed correctly for all $\theta \in \Theta$. 
				\item Otherwise dominance cut is dismissed incorrectly for all $\theta\in\Theta^{\epsilon}\triangleq\{\theta\in \Theta: J_{\eta}^*(\theta) \geq \bar{J}_{\eta}(\theta)\}$
			\end{enumerate}
		\end{enumerate}
		
		The most critical situation is in 2b), where we might explore descendants to $\eta$ in the certification algorithm, while these would be pruned in Algorithm 1 pointwise. 		
		The possible extra nodes in the analysis are added to $\mathcal{T}^k$ according to Assumption \ref{assumption-order-tree} and consequently, when a new node  is popped from the list in Step \ref{alg2_pop} of Algorithm \ref{alg2} to be analyzed, it is either a relevant one that would be solved in the online solver for a specific parameter in the corresponding region, or a redundant one which has been cut away in Algorithm \ref{alg1}. For parameters with redundant nodes where case 2b) in the proof of Theorem~\ref{theorem-sequence} occurs, we have that $J^*_{\eta}(\theta)\geq \bar{J}_{\eta}(\theta)$. According to Corollary \ref{cor:redundant-desc-exploration}, the extra exploration of descendants to $\eta$ will not affect the other exploration. Therefore, $ \hat{\mathbb{B}}_k(\theta) \supseteq \mathbb{B}(\theta)$, $\forall \theta\in \Theta_k$.
	\end{proof}
	 
 	\begin{theorem} \label{cor-conserve-cert}	
 	Assume Assumption \ref{assumption-order-tree} holds. 
 	Let $\kappa(\theta)$ be a complexity measure for solving the QP relaxation $\mathcal{P}(\theta)$ returned by $\textsc{QPcert} \left ( \mathcal{P}(\theta),\Theta \right )$ and that it satisfies $\kappa^j = \kappa(\theta)$, $\forall \theta\in \Theta^j \subseteq \Theta $, $\forall j$. Moreover, let $\{\Theta^i, \kappa^i_{tot} \}_{i=1}^{N_{\mathcal{F}}}$ be the collection of tuples in  $\mathcal{F}$ returned by Algorithm \ref{alg2} and let $\kappa^*_{tot}(\theta)$ be the total complexity for solving all QP relaxations encountered in Algorithm \ref{alg1} for $\theta \in \Theta_0 $. Then $\kappa^i_{tot} \geq \kappa^*_{tot}(\theta)$, $\forall \theta\in \Theta^i$.
 	\end{theorem}
 	\begin{proof} \label{proof-cor-conserve-cert} 
 	Based on the result from Theorem \ref{theorem-sequence}, the QP relaxations considered in Algorithm \ref{alg2} for a single point $\tilde{\theta}$ will be a superset of the QP relaxations that would be solved in Algorithm \ref{alg1}. This in turn means that the total complexity $\{\kappa^i_{tot}\}_{i}$ returned by Algorithm 2 will be an upper bound on the total complexity of Algorithm 1 (pointwise), i.e., $\kappa^i_{tot} \geq \kappa^*_{tot}(\theta)$, $\forall \theta\in \Theta^i$.   
 	\end{proof}

	\begin{remark} \label{remark:choise_kappa} 
		There is a freedom in the choice of complexity measure and Algorithm \ref{alg2} is not restricted to a specific certification measure. If the QP certification in \cite{arnstrom2021unifying} is employed in Step \ref{alg2_cert}, then the complexity measure $\kappa^j$ will be the iteration number, that is the total number of linear system equations that have to be solved. In other words, $\kappa^j$ will be the number of equality constrained QPs that an active-set method requires to solve to reach the optimal solution. 
		Another alternative is to fix $\kappa^j$ to $1$. Then the complexity measure will be the total number of nodes required to be explored in the \bnb tree, i.e., the number of relaxations (as in~\cite{axehill2010improved}). The other important choice is if $\kappa^j$ holds the number of floating point operations (flops) for solving mp-QPs, given by the QP certification algorithm. It then results in the total number of operations performed to solve all the relaxations in \bnb.
	\end{remark}	
	


        \section{Numerical Example} \label{sec:num-ex}
In this section, the proposed algorithm is tested in numerical experiments to illustrate its usefulness. To experimentally verify the analysis result, the upper bound on the complexity from Algorithm \ref{alg2} in a given $\theta$ is compared with the result of the online solver implementing Algorithm~\ref{alg1}. In the experiment, the complexity measure $\kappa(\theta)$ is chosen as the number of QP iterations. That is, for each value of $\theta\in\Theta_0$, the presented end result is the upper bound of the accumulated number of QP iterations from all relaxations necessary to solve in order to solve the mixed-integer problem for that particular value of $\theta$. In the experiments, the QP certification algorithm from \cite{arnstrom2021unifying} was used to certify the relaxations in Step \ref{alg2_cert} of Algorithm \ref{alg2}.
	
	For the evaluation, ten random mp-MIQPs have been generated with $n_c=2$, $n_b=4$, $m=6$, and $n_{\theta}=2$. The mp-MIQPs have the following form
	\begin{align*} \label{eq6}
	H &= \bar{H} \times \bar{H}^T, &  \quad f & \sim \mathcal{N}(0,1),  & \quad f_{\theta} & \sim \mathcal{N}(0,1), \\ 
	A &  \sim \mathcal{U}([0,1]), & \quad b &  \sim \mathcal{U}([1,2]),  & \quad  W & \sim \mathcal{U}([0,2]). 
	\end{align*}
	where $\bar{H} \sim \mathcal{N}(0,1)$, i.e., the Hessian matrices have been chosen such that they are symmetric and positive definite.  Furthermore, the parameter set considered is  $-1 \leq \theta \leq  1$. 
	
	Note that in the experiments, the conservative comparison is performed in the dominance cut in Algorithm \ref{alg2} and as a result the partition will be polyhedral. The computation time required for the complexity certification of each one of the ten random examples are listed in Table \ref{tab1}, when executed on an Intel$^{\circledR}$ Core 1.8 GHz i7-8565U CPU. Furthermore, the worst-case overall QP iteration count $\kappa_{max}$ of terminated regions  is provided in Table \ref{tab1}.

\begin{table}[htbp] 
	\caption{Results from a MATLAB implementation of Algorithm \ref{alg2} for 10 randomly generated mp-MIQPs.
	$\kappa_{max}$ is the worst-case accumulated iteration number,  
	and $t_{cert}$ denotes the computation time for the certification.}
	\label{tab1}
	\begin{center}
		\begin{tabular}{|c|c|c|}
			\hline
			$mp-MIQP$ & $t_{cert}$ &  $\kappa_{max}$  \\  
			\hline
			\# 1 & 12  & 27  \\ 
			\hline
			\# 2 & 231  & 41 \\ 
			\hline
			\# 3  & 97 & 51 \\ 
			\hline
			\# 4 & 332  & 51 \\ 
			\hline
			\# 5 & 143  & 46 \\ 
			\hline
			\# 6 & 173  & 38 \\ 
			\hline
			\# 7 & 399  & 48 \\ 
			\hline
			\# 8  & 572 & 47 \\ 
			\hline
			\# 9  & 327 & 47 \\ 
			\hline
			\# 10 & 785 & 58 \\ 
			\hline
		\end{tabular}
	\end{center}
\end{table}

	The resulting partitioning of the parameter space based on the total number of iterations for one of the examples of random mp-MIQPs is shown in Fig.~\ref{fig1}, where parameters in a region that share the same number of iterations are illustrated using the same color.
	
	\begin{figure}[htbp] 
		\centerline{\includegraphics[scale=0.36]{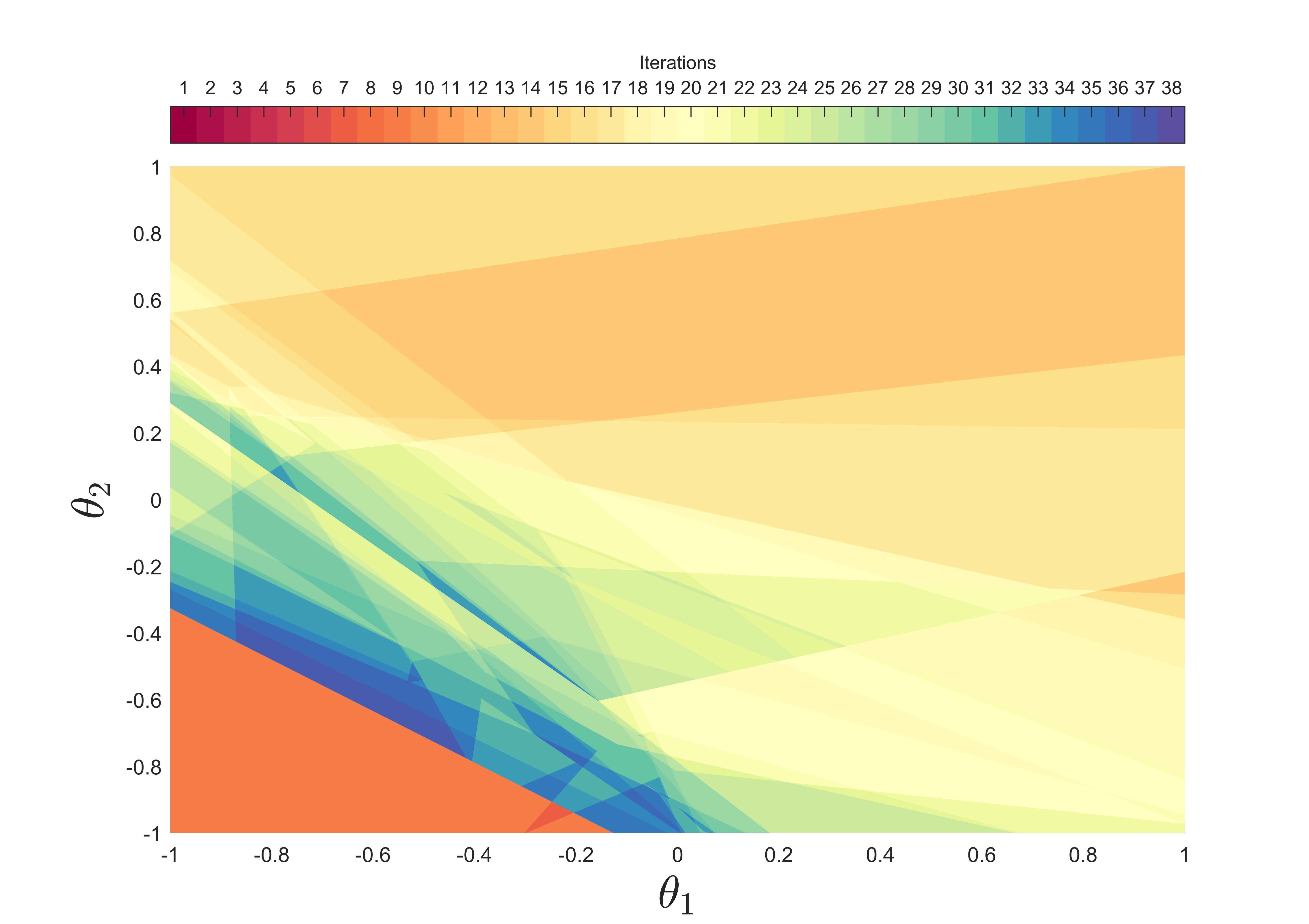}}
		\caption{Results from Algorithm~\ref{alg2} applied to a random example. The colors illustrate the accumulated number of QP iterations for a specific $\theta$.}
		\label{fig1}
	\end{figure}
	
	In order to verify the result from the certification, the worst-case iteration bounds obtained by the presented method were compared with results from Monte-Carlo (MC) simulations, i.e., applying Algorithm \ref{alg1} to the same MIQP problems \eqref{eq5} for some specific parameters 
	from the parameter set. To have the same complexity measure also online, the dual-active set method in \cite{arnstrom2021efficient}  
	was applied to solve the QP relaxations in Step \ref{alg1_solv} of Algorithm \ref{alg1}, and the total number of QP iterations performed was summed. 
	
	Fig.~\ref{fig2} demonstrates the result of solving the online MIQP problem for 10000 samples on a deterministic grid in the parameter space, for the same example as in Fig.~\ref{fig1}. The Chebychev centers of regions in Fig.~\ref{fig1} were also included as sample points to guarantee that there was at least one sample from each region. As the figures indicate, the result from the complexity certification coincides with the online algorithm in all sample points, despite that the conservative upper bound is used in Algorithm \ref{alg2}. The conclusions from the experiments are that the quality of the upper bound is clearly practically useful and the bound seems to be very close to tight.
	
	\begin{figure}[htbp] 
		\centerline{\includegraphics[scale=0.4]{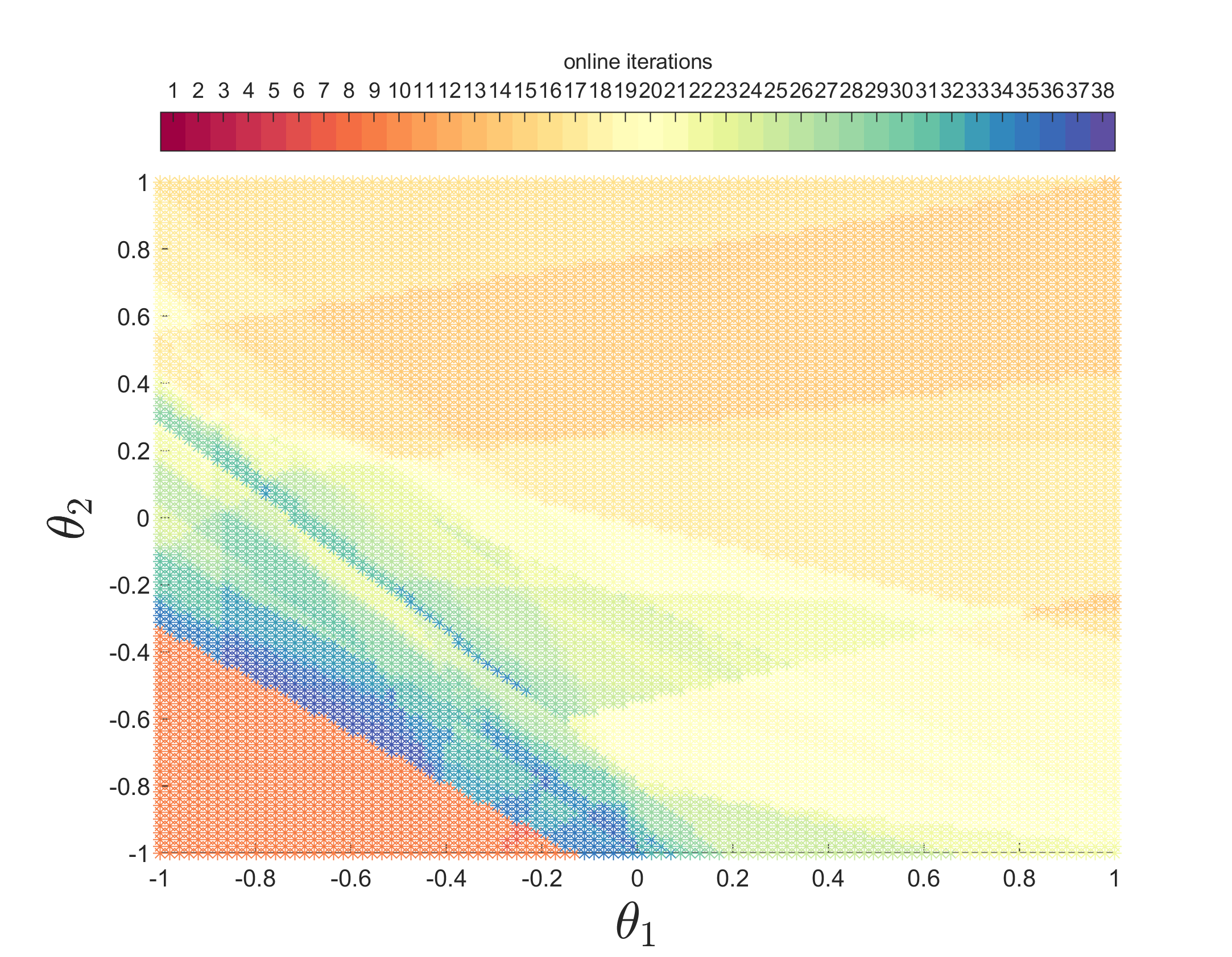}}	
		\caption{The total QP iteration number for 10000 samples specified by * in the parameter space derived by applying online MIQP to the same example as in Fig.~\ref{fig1}. Points with the same color share the same number of complexity numbers.}
		\label{fig2}
	\end{figure}	
	
	%
	To illustrate the result of the certification method for higher dimensional problems, random examples were generated with $8$ binary decision variables and $8$ constraints. The remaining variables were defined as before with the appropriate dimensions. 
	Algorithm \ref{alg2} was applied to certify the random MIQPs. The maximum total iteration count in the worst region here was recorded as $\kappa_{max} = 126$. 	
	To verify the result, the online MIQP problem was solved in $10000$ sample points on a deterministic grid in the parameter space. The minimum difference between the certification results of the analysis and the online algorithm in the sample points was $0$, verifying that the analysis would never underestimate the complexity. 
	In $7$ percent of analyzed points 
	there were non-zero variations though, in which the maximum difference was 22.	
	The difference originates from the dominance cut condition in Step \ref{alg2_comp} of Algorithm \ref{alg2}, which as described in Section \ref{subsec:spatial-decom}, is somewhat conservative. Note that the result from the proposed certification would be exact if the comparison was done exactly. 
	
	All numerical experiments were implemented in MATLAB and Gurobi 9.0.2 \cite{gurobi} 
	was used to solve the potentially indefinite QP optimization problem in Step \ref{alg2_comp} of Algorithm \ref{alg2}.

		\section{Conclusion} 
In this paper, complexity certification of a standard branch and bound method for MIQPs has been addressed. An algorithm for computing a useful upper bound of the worst-case computational complexity for solving any possible MIQP that can arise from a specific parameter in a polyhedral parameter set has been presented. Compared to the previous work, the current work uses recent algorithms for exact complexity certification of active-set QP methods, enabling also taking into account the complexity originating from solving the relaxations in the nodes. Even though the main focus in this work has been a complexity measure in terms of accumulated number of QP iterations, there is a freedom in this choice and alternatives such as flops and the size of the search tree can in principle be considered. Results from numerical experiments illustrate that the result is in general conservative but still useful. In fact, in many of the considered examples, the computed bound is tight. In future work, MILPs will be considered, the details of the exact non-conservative application of the dominance cut condition will be studied, and more general branching-rules will be considered.

		
        \section*{Acknowledgment}
This work was partially supported by the Wallenberg AI, Autonomous
Systems and Software Program (WASP) funded by the Knut and Alice
Wallenberg Foundation.

        
		\bibliography{IEEEabrv,references} 

\end{document}